\documentclass[a4paper]{article}

\usepackage{amssymb,amsfonts,amsmath,amsthm}
\usepackage{tabularx}
\usepackage{graphicx}
\usepackage{epstopdf}
\usepackage{verbatim}
\usepackage{algorithm} 
\usepackage{algorithmic} 
\usepackage{color}

\usepackage{graphicx}
\usepackage{tikz}
\usepackage{authblk}

\newtheorem{Theorem}{Theorem}

\newtheorem{Lemma}{Lemma}
\newtheorem{Definition}{Definition}

\title{How to select the largest $k$ elements from evolving data?\footnote{This work was partially supported by NSFC.}}
\author{Qin Huang}
\author{Xingwu Liu}
\author{Xiaoming Sun}
\author{Jialin Zhang}
\affil{Institute of Computing Technology, Chinese Academy of Sciences}

\date{}

\begin{document}

\maketitle

\vspace{-0.2cm}
\textbf{Abstract.}{
In this paper we investigate the top-$k$-selection problem, i.e. determine the largest, second largest, ..., and the $k$-th largest elements, in the dynamic data model. In this model the order of elements evolves dynamically over time. In each time step the algorithm can only probe the changes of data by comparing a pair of elements. Previously only two special cases were studied~\cite{Anagnostopoulos}: finding the largest element and the median; and sorting all elements. This paper systematically deals with $k\in [n]$ and solves the problem almost completely. Specifically, we identify a critical point $k^*$ such that the top-$k$-selection problem can be solved error-free with probability $1-o(1)$ if and only if $k=o(k^*)$. A lower bound of the error when $k=\Omega(k^*)$ is also determined, which actually is tight under some condition. On the other hand, it is shown that the top-$k$-set problem, which means finding the largest $k$ elements without sorting them, can be solved error-free for all $k\in [n]$. Additionally, we extend the dynamic data model and show that most of these results still hold.
}

\section{Introduction}
Algorithms, in the classical scenarios, assume the input is given at the beginning of the computation. In the era of big data, mass data force the algorithms to deal with changing data. Online algorithms, streaming algorithms are traditional ways to characterize such kind of data where these algorithms care more about the new coming data. In this paper, we consider another interesting paradigm which cares about the dynamically changing  relationship among the existing data. Such scenarios happen a lot in the real life: in the internet, besides new webpages are created everyday, the relative importance between two old webpages are also rapidly changing; in the market, the relative popularity between two similar goods evolves dynamically due to the changing price, advertisement strategy, etc; in the social network, the rank of hot topics is also changing due to the intrinsic shift in public opinion. More importantly, the order among objects (webpages, goods, topics) is more significant than the objects themselves. Many applications care about such order, especially the objects with higher order, like search engine, goods display arrangement, or guidance of public opinion. There two methodologies are widely used to learn the order: rating vs. ranking.  Researches may recruit volunteers or train models to either grade each objects or judge the comparison between two or more objects. The advantage and weakness between these two methods is a controversial question for a long time~\cite[Ch.7]{Thurstone}.

In this paper, we focused on a simplified model based on evolving relationships among data and learning methods based on comparison. The framework we use was first proposed by Anagnostopoulos et al.~\cite{Anagnostopoulos} to study many traditional problems including sorting and selection. In the model, a set $U$ of $n$ objects is equipped with a total order $\pi^t$ that evolves over discrete time $t\in \mathbb{N}$. At each time slot, $\alpha$ random consecutive pairs are swapped. In this paper, we further generalize their model to Gaussian-swapping model by relaxing the locality condition that we only swap consecutive pairs in the evolution. The evolution is stochastic, and the only way to probe $\pi^t$ is by querying: at time $t$, one can choose an arbitrary pair of objects in $U$ and query their order in $\pi^t$. The goal is to learn the information about the true order $\pi^t$ at every time $t$. In this paper, we focus on the top-$k$ largest objects in the order. See the formal definition of the problem in Section~\ref{sec:pre}.

Anagnostopoulos et al. \cite{Anagnostopoulos} proposed a randomized sorting algorithm for the whole set, which at every time $t$, with probability $1-o(1)$, guarantees that the {\em Kendall tau} distance between $\pi^t$ and the predicted order $\tilde{\pi}^t$ is at most $O(n\ln\ln{n})$. On the other hand, they showed that there is only a minor opportunity to improve the algorithm since the error is lower-bounded by $\Omega(n)$, for $t>\frac{n}{8}$. They also considered selecting the $k$-th largest elements (particularly the largest element and the median) problem and showed that with probability $1-o(1)$, their selection algorithm will output the right element, for any $k\in [n]$.
In the statistical data setting, when a selection algorithm output the $k$-th largest element, it automatically determined the set of largest $k$ elements (see for example Knuth's book~\cite{Knuth}). However, this is not apparent in this evolving data setting. A more challenge problem is that besides the set of the largest $k$ elements, the algorithm is expected to output the right order of these elements at the same time. We call the first problem {\em top-$k$-set problem} and the latter one {\em top-$k$-selection problem}.



\vspace{0.2cm}
\noindent {\bf Our contributions}

We firstly summarize our results for Anagnostopoulos's model.  In this paper, we provide an algorithm which for every time $t$, it can predict the top-$k$-set (i.e. the largest $k$ elements) with probability $1-o(1)$. For top-$k$-selection problem, we show almost tight results under the probability $1-o(1)$. In the work of Anagnostopoulos et al.~\cite{Anagnostopoulos}, they solved the case of $k=1$  and $k=n$ where the first case can be correctly output with probability $1-o(1)$, while the error lower bound is $\Omega(n)$ under the probability $1-o(1)$ for the latter case. In this paper we solve the rest cases. For constant $\alpha$, we propose an algorithm which aims to minimize the Kendall tau distance between the largest $k$ elements in $\pi^t$ and our predication $\tilde{\pi}^t$. When $k=o(\sqrt{n})$, our algorithm can guarantee that with probability $1-o(1)$ it will output the largest $k$ elements with the right order, i.e. the Kendall tau distance is 0. For $k=\Theta(\sqrt{n})$, the probability to output exactly right order is still a constant, while we also show that no algorithm can achieve success probability $1-o(1)$ in this case. For $k=\omega(\sqrt{n})$, our algorithm can guarantee that with probability $1-o(1)$ Kendall tau distance error is at most $O(\frac{k^2}{n})$. On the other hand, no algorithm with successful probability $1-o(1)$ can achieve better than $\Theta(\frac{k^2}{n})$. We then generalize our algorithm to general $\alpha$. All positive results and some of the negative results can be extended to any $\alpha=o(\frac{\sqrt{n}}{\ln n})$, see Section~\ref{sec:korderC} for details. We summarize our results in Table~\ref{tab:consecutive}.

\begin{table}[!h]
\centering
\caption{\emph{Consecutive-swapping model}}\label{tab:consecutive}
\begin{tabular}{ | p{1.5cm}<{\centering} | p{4.5cm}<{\centering} |}
  \hline
   $k$ & $X \triangleq \mathrm{KT}(\tilde{\pi}_k^t,\pi_k^t)$    \\
  \hline
   $o(\sqrt{\frac{n}{\alpha}})$ & $\Pr(X=0)=1-o(1)$*    \\
  \hline
   $\Theta(\sqrt{\frac{n}{\alpha}})$& $\Pr(X=0)=\Omega(1)$    \\
  \hline
  $\omega(\sqrt{\frac{n}{\alpha}})$ & $\Pr(X=O(\frac{k^2\alpha}{n}))=1-o(1)^{\dag}$ \\
  \hline
\end{tabular}

\noindent In $*$ case these bounds are tight among any algorithm with success probability $1-o(1)$; In $\dag$ case, the tight result holds for constant $\alpha$. See section~\ref{sec:korderC}
\end{table}


We then generalize our results to the Gaussian-swapping model. In the Gaussian-swapping model, at each time slot, the rank difference of any chosen swapped pairs follows a Gaussian distribution, while in the original model, the rank difference is always $1$ (i.e. consecutive pairs). We get similar results for this model and summarize them in Table~\ref{tab:Guassian}.

\vspace{-0.3cm}
\begin{table}[!h]
\centering
\caption{\emph{Gaussian-swapping model}}\label{tab:Guassian}
\begin{tabular}{ | p{1.5cm}<{\centering} | p{4.5cm}<{\centering}  |}
  \hline
   $k$ & $X\triangleq \mathrm{KT}(\tilde{\pi}_k^t,\pi_k^t)$   \\
  \hline
   $o(\frac{\sqrt{n}}{\ln^{0.25} n})$ & $\Pr(X=0)=1-o(1)$      \\
  \hline
   $\Theta(\frac{\sqrt{n}}{\ln^{0.25} n})$ &$\Pr(X=0)=\Omega(1)$    \\
  \hline
  $\omega(\frac{\sqrt{n}}{\ln^{0.25} n})$ & $\Pr(X=O(\frac{k^2\ln n}{n}))=1-o(1)$   \\
  \hline
\end{tabular}
\end{table}

\noindent {\bf Related work}

The investigation of top-$k$-set and top-$k$-selection problem can date back to 1960s~\cite{Kislitsyn,Knuth} and have been studied extensively~\cite{BFPRT,Yap,Zwick,Arnold}. Sorting-related problems on dynamic data were just initiated very recently. Anagnostopoulos et al.~\cite{Anagnostopoulos} focused on two extreme cases ($k=1$ and $k=n$) of the top-$k$-selection problem, 
which inspired us to explore the problem for general $k$. Additionally, the data evolving model in \cite{Anagnostopoulos} is restricted in the sense that at every time step $t$, only $\alpha=O(1)$ consecutive pairs are swapped. In this work we extend the model to allow $\alpha=\omega(1)$ and the pairs to be non-consecutive, so that the data can evolve faster and less locally. 

Dynamic data are also studied in the graph setting~\cite{Anagnostopoulos2012}. They considered two classical graph connectivity problems (path connectivity and minimum spanning trees) where the graph keeps changing over time and the algorithm needs to periodically probe the graph to maintain the path or spanning tree. In dynamic graph model, more results are proposed recently to show new ideas and experimental results.  Bahmani et al.~\cite{Bahmani} designed an algorithm to approximately compute the PageRank, and Zhuang et al. \cite{Zhuang} considered the influence maximization problem in dynamic social networks. On the other hand, Labouseur et al. \cite{Labouseur} and Ren \cite{Ren} dealt with the data structure and management issues which enable efficient query processing for dynamic graphs.
Moreland \cite{Moreland} experimentally verified the theoretical bounds predicted in~\cite{Anagnostopoulos}.

It is worth noting that our dynamic data model looks quite different point of view compared with computing under noisy information \cite{Ajtai,Feige}. In the research of noisy data, the key difficulty is the misleading information. But in our model, we treat the dynamical evolution and query results to be correct, while the difficulty comes from the fact that we cannot immediately learn all changes. We can only observe the real data by probing the partial information about the data. The key is to choose the probing strategies in order to approximate the real data with high probability.

In the algorithm community, there are many other models dealing with dynamic and uncertain data, from various points of view. However, none of these captures the two crucial aspects of our dynamic data scenario: the underlying data keeps changing, and the data exposes, by the probe model, only  limited information to the algorithm. For example, data stream algorithms \cite{Babcock} deal with a stream of data, typically with limited space, but algorithms can observe the entire data; local algorithms on graphs \cite{Bressan,Fujiwara} try to capture a certain property using a limited number of query to the underlying graphs, but the graphs are typically static; in online algorithms \cite{Albers}, though the data is coming over time and be processed without knowledge of the future data, the algorithms know all the data up to now; the multi-armed-bandit model \cite{Kuleshov} tends to optimize the total gain value in a finite exploration-exploitation process, while our framework concerns the performance of the algorithm at every time slot.

The rest of the paper is organized as follows. In Section~\ref{sec:pre}, we provide the formal definition of our models and the problems. Section~\ref{sec:main} is devoted to solving top-$k$-set problem and top-$k$-selection problem in the consecutive-swapping model. In Section~\ref{sec:Guassian}, the problems are investigated in the Gaussian-swapping model. Section~\ref{sec:con} concludes the paper.

\section{Preliminaries}\label{sec:pre}

We now present the formal definition of our evolving model of dynamic elements. Let $U$ be a set with $n$ elements $U=\{u_1,...,u_n\}$, and let $\mathcal{U}$ be the set of all possible total orders over $U$, that is, $\mathcal{U}=\{\pi: U\rightarrow \{1,2,\cdots,n\}\ |\ \pi(u_i)\neq \pi(u_j) \forall i\neq j\}$. For convenience, related to order $\pi$, we say $\pi^{-1}(1)$ is the largest element, and $\{\pi^{-1}(1),\cdots, \pi^{-1}(k)\}$ is the largest $k$ elements.

In the paper, we consider the evolving process that the order of $U$ will gradually change as time goes by. Let $\pi^t\in \mathcal{U}$ be the total order of $U$ at time step $t$. For every $t>1$, $\pi^t$ is obtained from $\pi^{t-1}$ by sequential $\alpha$ steps. In each step, we swap one random pair of consecutive elements. Here, $\alpha$ can be a constant or an integer function of $n$. We call such dynamic model \emph{consecutive-swapping model}. We further generalize the model to \emph{Gaussian-swapping model} by relaxing the locality condition that we only swap consecutive pairs in the evolution. We define the distance of $u_i\in U$ and $u_j\in U$ related to a total order $\pi\in \mathcal{U}$ as $d_{\pi}(u_i,u_j)=|\pi(u_j)-\pi(u_i)|$. In the Gaussian-swapping model, for every $t>1$, $\pi^t$ is still obtained from $\pi^{t-1}$ by sequential $\alpha$ steps, and in each step, one random pair of elements (not necessary consecutive) will be swapped according to the following distribution. We firstly choose $D$ according to Gaussian distribution, that is , $\Pr(D=d) = \beta e^{\frac{-d^2}{2}}$ where $\beta$ is the normalizing factor. And then uniformly randomly choose a pair with distance $D$. Thus, the probability to choose the pair $(u_i,u_j)$ is $\frac{\beta e^{\frac{-d^2}{2}}}{n-d}$ where $d$ is the distance of $u_i$ and $u_j$ related to the current total order.

Algorithms in this paper are assumed to have restricted access to the dynamic data. The only way the algorithm can probe the data is by comparative queries: at any time t, given a pair of elements $u_1,u_2\in U$, it can obtain whether $\pi^t(u_1)>\pi^t(u_2)$ or not. An algorithm can query at most one pair of elements at one time step.

Now we define $\mathcal{I}$-sorting problem for any index set $\mathcal{I}\subseteq \{1,\cdots,n\}$: find out not only all the elements whose ranks belong to $\mathcal{I}$, but also their ranks respectively. The concept of $\mathcal{I}$-sorting problem unifies both the sorting problem ($|\mathcal{I}|=n$) and the selection problem ($|\mathcal{I}|=1$). This paper mainly studies \emph{Top-$k$-selection  problem}, which is the special case of $\mathcal{I}$-sorting problem with $\mathcal{I}=\{1,2,...k\}$ for $1\leq k\leq n$. For simplicity, we use notation $\pi_k^t$ to represent the order restricted to largest $k$ elements in this paper. A closely-related problem, called \emph{Top-$k$-set problem}, is also studied. It requires to find out $\{(\pi^t)^{-1}(1),\cdots,(\pi^t)^{-1}(k)\}$ at every time t, not caring about the ranks.

We then define the measurement of performance of the algorithm.  In top-$k$-set problem, we want to maximize the probability that our output set is exactly the same as the true set for any sufficiently large $t$. In top-$k$-selection problem, besides guarantee that the top-$k$-set output by the algorithm is correct, we also try to minimize the Kendall tau distance between our output order and the true order. Suppose $\tilde{\pi}_k$ is the output of our algorithm and $\pi_k$ is the true order of largest $k$ elements. The Kendall tau distance $\mathrm{KT}(\tilde{\pi}_k,\pi_k)$ is defined as $\mathrm{KT}(\tilde{\pi}_k,\pi_k)=|\{(a,b): a<_{\tilde{\pi}_k}b \wedge b <_{\pi_k} a\}|$ where  $a<_{\pi} b$ if $a$ is smaller than $b$ according to permutation $\pi$. Note that the maximum Kendall tau distance is $\Theta(k^2)$.

Throughout our paper, one building block of our algorithms is the randomized quick-sort algorithm. We describe the randomized quick-sort algorithm briefly. It works as follows: $(1)$ Randomly and uniformly pick an element, called a pivot, from the array. $(2)$ Compare all elements with the pivot. Two sub-arrays are obtained: one consisting of all the elements smaller than the pivot, and the other consisting of the other elements except the pivot.  $(3)$ Recursively apply steps 1 and 2 to the two sub-arrays until all the sub-arrays produced are singleton.

\vspace{-0.2cm}
\section{Consecutive-swapping Model}\label{sec:main}
\vspace{-0.1cm}
In this section, we consider the top-$k$-set problem and the top-$k$-selection problem in the \emph{consecutive-swapping model}. For the top-$k$-set problem, Subsection~\ref{sec:ksetC} shows an algorithm which is error-free with probability $1-o(1)$ for arbitrary $k$. Subsection~\ref{sec:korderC} is devoted to the top-$k$-selection problem. It presents an algorithm that is optimal when $\alpha$ is constant or $k$ is small.

\vspace{-0.2cm}
\subsection{An algorithm for the \emph{Top-$k$-set} problem}\label{sec:ksetC}
\vspace{-0.2cm}
The basic idea is to repeatedly run quick-sort over the data $U$, extract the set of the largest $k$ elements by the resulting order, and output this set during the next run. But an issue should be addressed: since the running time of quick-sort is $\Omega(n\ln n)$, the set of the largest $k$ elements will change with high probability during the next run, leading to out-of-date outputs. Considering that the rank of every element does not change too much during the next quick-sort, a solution is to parallel sort a small subset of $U$ that contains the largest $k$ elements with high probability.

Specifically, the algorithm \emph{Top-$k$-set} consists of two interleaving algorithms (denoted by $QS_1$ and $QS_2$, respectively), each of which restarts once it terminates. In the odd steps, $QS_1$ calls quick-sort to sort $U$, preparing two sets $L$ and $C$. Actually, with high probability, $L$ consists of the elements that will remain among top-$k$ during the next run of $QS_1$, while $C$ contains all the elements that can be among top-$k$ in the period. $QS_2$ uses the sets $L$ and $C$ computed by the last run of $QS_1$ to produce the estimated set of largest $k$ elements. At time $t$, the output $\widetilde{T}_t$ of the algorithm is the set $\widetilde{T}$ computed by the previous run of $QS_2$.

\vspace{-0.2cm}
\begin{algorithm}[H]
 \textbf{Input:} A set of elements $U$\\
 \textbf{Output:} $\widetilde{T}$
 \begin{algorithmic}[1]
 \WHILE{(true)}
    \STATE \textbf{Execute in odd steps:}    /*$QS_1$*/
    \STATE $\tilde{\pi}\leftarrow $ quick\_sort($U$) in decreasing order
    \STATE $L\leftarrow\tilde{\pi}^{-1}(\{1,2,...,k-c\alpha\ln n \}) $ and $C\leftarrow\tilde{\pi}^{-1}(\{k-c\alpha\ln n +1,..., k+c\alpha\ln n\})$ /*The constant $c$ will be determined in the proof*/
    \STATE \textbf{Execute in even steps:}    /*$QS_2$*/
    \STATE\label{line:quicksortC}
          $\tilde{\pi}_C \leftarrow $  quick\_sort($C$) in decreasing order
    \STATE $\widetilde{T}\leftarrow L \bigcup\tilde{\pi}_C^{-1}(\{1,2,...,c\alpha\ln n\})$
 \ENDWHILE
 \caption{\emph{Top-$k$-set} } \label{fig:topk}
 \end{algorithmic}
\end{algorithm}

Two lemmas are needed for analyzing the performance of Algorithm \ref{fig:topk}.
\def\quicksort{
For any $\alpha=o(n)$, the running time of the standard randomized quick-sort algorithm in the \emph{consecutive-swapping model} is $O(n\ln{n})$ in expectation and with probability $1-o(1)$.
}
\begin{Lemma} \label{quicksort}
\quicksort
\end{Lemma}

\begin{proof}
This proof is inspired by the proof of \cite[Proposition 3]{Anagnostopoulos} and the proof of the time complexity of the randomized quick-sort algorithm\cite{D.P.}.

We first recall the main steps in proving that the time complexity of the randomized quick-sort on static data is $O(n\ln n)$ with probability $1-o(1)$.
\begin{enumerate}
  \item A pivot is said to be good if it divides the array of size $s$ into two sub-arrays each having at least $\gamma s$ elements, where $0<\gamma<0.5$ is a constant. Thus the number of good pivots along a given path from the root to a leaf in the quick-sort execution tree is $O(\ln n)$.
  \item In a given path, a pivot is good with constant probability and bad (i.e., not good) also with constant probability. By using Chernoff bound, the length of a given path is $O(\ln n)$ with probability at least $1-o(1)$.

  \item By union bound, the lengths of all the paths are $O(\ln n)$ with probability $1-o(\frac{1}{n})$. Therefore, the running time is $O(n\ln n)$ with probability $1-o(1)$ and in expectation.
\end{enumerate}

In the \emph{consecutive-swapping} model, since the true order keeps evolving, a pivot that is good at the time it is chosen might divide the array of size $s$ into two parts one of which has fewer than $rs$ elements. However, with probability $1-O(\frac{\alpha}{n})$, each part will contain at least $\frac{\gamma s}{2}$ elements. As a result, in the \emph{consecutive-swapping} model, we redefine the term \emph{good pivot} to be a pivot which divides the corresponding array of size $s$ into two sub-arrays both having at least $\frac{\gamma s}{2}$ elements. If $\alpha=o(n)$, with at least a constant probability $\beta$ that a pivot is good. In any given path of length $L$, the number of good pivots is $O(\ln n)$ and the expectation of bad pivots is at most $(1-\beta) L$. Thus by Chernoff bound the number of bad pivots is also $O(\ln n)$ with probability $1-o(1)$. Therefore, steps 2 and 3 hold.

\end{proof}

\def\range{
In the \emph{consecutive-swapping model} with $\alpha=o(n)$, consider a run of the standard randomized quick-sort from time $t_0$ to $t_1$. For any $u_i\in U$, the number of incorrectly ordered pairs $(u_i,u_j)$ is $O(\alpha\ln{n})$ with probability $1-O(\frac{1}{n^3})$. Specifically, with probability $1-O(\frac{1}{n^2})$, $|\pi^{t_1}(u_l)-\tilde{\pi}^{t_1}(u_l)|=O(\alpha\ln{n})$ for any $l$.
}
\begin{Lemma}\label{range}
\range
\end{Lemma}

 \begin{proof}
 This proof is similar to the proof of \cite[Lemma 6]{Anagnostopoulos}.

Arbitrarily fix a $u_i\in U$. We first consider the set $S=\{u_j\in U|u_i<_{\pi^{t_1}} u_j,u_j<_{\tilde{\pi}^{t_1}} u_i\}$.
Of course, $S=S_1\bigcup S_2$, where $S_1=\{u_j\in U|\exists t\in [t_0, t_1),u_j<_{\pi^t} u_i\}$ and $S_2=\{u_j\in U|\forall t\in [t_0, t_1],u_i<_{\pi^t} u_j\}$.

Let $Z_1=|S_1|$. It is bounded by the number of swaps involving $u_i$. Since at each time step, $u_i$ was chosen to swap with probability at most $\frac{2\alpha}{n}$ and the running time is bounded by $O(n\ln n)$ in expectation and with high probability $1-o(1)$. We have that $E[Z_1]\leq 2c_1\alpha\ln n$ for some constant $c_1$. And by Chernoff bound, $\Pr(Z_1 \leq 4c_1\alpha\ln n)=1-O(\frac{1}{n^3})$.

For any $u_j\in S_2$, there exists a pivot $u_l$ which incorrectly places $u_i$ and $u_j$. Specifically, when $u_l$ is compared with $u_i$, $u_l<u_i$ (by the true order at that time), but when it is compared with $u_j$, $u_j < u_l$.

Let $Z_2=|S_2|$. It is thus bounded by the number of swaps that involve any pivot which is compared with $u_i$. We let $X_l$ be the number of steps that an element $u_l$ is a pivot and $Y_l$ be the number of swaps that involve $u_l$ while it is a pivot. Let $Q_i$ be the set that consists of all the pivots which are compared with $u_i$. By the properties of the quick-sort algorithm, we have
 $$E[\sum_{l:u_l\in Q_i} X_l]=O(n)$$ and
 $$\sum_{l:u_l\in Q_i} X_l\leq c_1n\ln n$$
 with probability $1-o(1)$ for a constant $c_1$. Since $Y_l \sim Binomial(X_l,2\alpha/n)$, we have $E[Z_2]\leq E[E[\sum_{u_l\in Q_i}Y_l|X_l]]=O(\alpha)$. By Chernoff bound, $\Pr(Z_2\leq 4c_1\alpha\ln n)=1-O(\frac{1}{n^3})$.

Altogether, $\Pr(|S|\leq 8 c_1\alpha\ln n)=1-O(\frac{1}{n^3})$. Likewise, we also have $\Pr(|S'|\leq 8 c_1\alpha\ln n)=1-O(\frac{1}{n^3})$, where $S'=\{u_j\in U|u_j<_{\pi^{t_1}} u_i,u_i<_{\tilde{\pi}^{t_1}} u_j\}$. The first part of the lemma thus holds.

Now prove the second part of the lemma. By union bound, for every $l$, $|\pi^{t_1}(u_l)-\tilde{\pi}^{t_1}(u_l)|\leq c\alpha\ln n$ holds with probability at least $1-O(\frac{1}{n^2})$, where $c=16 c_1$.
 \end{proof}

The following theorem shows that Algorithm \ref{fig:topk} is always error-free, up to high probability.
\def\topp{
Assume that $\alpha=o(\frac{\sqrt{n}}{\ln{n}})$. For any $1\leq k\leq n$, we have $Pr(\widetilde{T}_t=(\pi^t)^{-1}(\{1,2,...,k\}))=1-o(1)$, where $\widetilde{T}_t$ is the output of Algorithm \emph{Top-$k$-set} at time $t$, $\pi^t$ is the true order of $U$ at time $t$, and $t$ is such that the algorithm $QS_1$ has run at least once.
}
\begin{Theorem}  \label{The:theoremone}
\topp
\end{Theorem}
\begin{proof}

Consider a run of $QS_1$, which starts at $t_0$ and ends at $t_1$. By Lemma~\ref{quicksort}, we have $t_1-t_0=O(n\ln{n})$ in expectation and with probability $1-o(1)$. By Lemma~\ref{range}, we have that $|\pi^{t_0}(u_l)-\tilde{\pi}^{t_0}(u_l)|\leq c_1\alpha\ln{n}$ with probability $1-O(\frac{1}{n^2})$ for every $u_l\in U$ and some constant $c_1$. During $[t_0,t_1]$, the rank of any element $u_l$ changes less than $c_2\alpha\ln{n}$ with probability $1-O(\frac{1}{n^2})$ for some constant $c_2$. Hence, we have $|\pi^t{(u_l)}-\tilde{\pi}^{t_0}(u_l)|\leq (c_1+c_2)\alpha\ln{n}$ with probability $1-o(1)$ at any time $t\in [t_0,t_1]$ and any $u_l\in U$. Letting $c=c_1+c_2$, we have $|\pi^t{(u_l)}-\tilde{\pi}^{t_0}(u_l)|\leq c\alpha\ln{n}$ with probability $1-o(1)$.

Note that $L$ contains all the elements $u_i$ such that $\tilde{\pi}^{t_0}(u_i)\leq k-c\alpha\ln{n}$. Then for any $t\in[t_0,t_1]$ and any $u_i\in L$, we have that $\pi^t(u_i)\leq k$ with probability $1-o(1)$. 
 Consider the set $R=U\backslash (L \bigcup C)=\{u_l: \tilde{\pi}^{t_0}(u_l)\geq k+c\alpha\mathrm{ln}n+1\}$. Then for any $t\in[t_0,t_1]$ and any $u_i\in R$, we have that $\pi^t(u_i)> k$ with probability $1-o(1)$. Therefore, $L\bigcup C$ contains all the elements whose true rank is no more than $k$ during $[t_0,t_1]$, with probability $1-o(1)$.

In line~\ref{line:quicksortC}, quick\_sort($C$) requires time $O(\alpha\ln{n}(\ln \alpha+\ln\ln n))$ with probability $1-o(1)$ (by Lemma~\ref{quicksort}). If the element of rank $k$(say $\tilde{u}$) at the beginning of line~\ref{line:quicksortC} does not swap during the execution of line~\ref{line:quicksortC}, the algorithm can always return the correct set of the largest $k$ elements at the end of quick\_sort($C$). 
This is because by the \emph{consecutive-swapping model}, if $\tilde{u}$ does not swap, the sign of $\pi^t(\tilde{u})-\pi^t(v)$ does not change for any $v\in U$. And the probability that $\tilde{u}$ swaps during the execution of line~\ref{line:quicksortC} is bounded by $O(\frac{\alpha^2\ln{n}(\ln\alpha+\ln\ln n)}{n})$. Hence, if $\alpha=o(\frac{\sqrt{n}}{\ln{n}})$, the algorithm can always return the correct set of the largest $k$ elements with probability $1-o(1)$. During the next round of sorting $C$(before a new set of the largest $k$ elements is computed), the set of the largest $k$ elements can remain the same with probability $1-O(\frac{\alpha^2\ln{n}(\ln\alpha+\ln\ln n)}{n})=1-o(1)$, implying the result.
\end{proof}

\vspace{-0.5cm}
\subsection{Algorithms for \emph{Top-$k$-selection}}\label{sec:korderC}
\vspace{-0.1cm}

Now we present an algorithm to solve the top-$k$-selection problem. The basic idea is to repeatedly run quick-sort over the data $U$, extract a small subset that includes the elements that can be among largest $k$ elements during the next run. To identify the exact largest $k$ elements together with their order, the small set is sorted in parallel and the order of the largest $k$ elements is produced accordingly. Like in designing the top-$k$-set algorithm, there is also an issue to address: since sorting the small set requires time $\Omega(k)$, the order of the largest $k$ elements will soon become out of date. Again note that the rank of every element does not change too much (actually, upper bounded by a constant) during sorting the small set, so the order of the largest $k$ elements can be regulated locally and maintain updated.

Specifically, Algorithm \ref{topkorder} consists of four interleaving algorithms ($QS_1$, $QS_2$, $QS_3$, and \emph{Local-sort}), each of which restarts once it terminates. At the $(4t+1)$-th time steps, $QS_1$ invokes a quick-sort on $U$, preparing a set $C$ of size $k+O(\alpha \ln n)$ which contains all the elements that are possible to be ranked among top-$k$ during the next run of $QS_1$ with high probability. At the $(4t+2)$-th time steps, $QS_2$ calls another quick-sort on the latest $C$ computed by $QS_1$, producing a set $P$ of size $k$. The set $P$ exactly consists of the largest $k$ elements of $U$ during the next run of $QS_2$ with high probability. At the $(4t+3)$-th time steps, the other quick-sort is invoked by $QS_3$ on the latest $P$ computed by $QS_2$, periodically updating the estimated order over $P$. The resulting order is actually close to the true order over $P$ during the next run of $QS_3$. Finally, at the $(4t)$-th time steps, an algorithm \emph{Local-sort} is executed on the estimated order over $P$ produced by the last run of $QS_3$, so as to correct the local errors of the order. At any time $t$, the output $\tilde{\pi}_k^t$ of Algorithm \ref{topkorder} is the last $\tilde{\pi}_k$ computed by \emph{Local-sort}.

\vspace{-0.2cm}
\begin{algorithm}
  \textbf{Input:} A set of elements $U$  \\
  \textbf{Output:} $\tilde{\pi}_k$
  \begin{algorithmic}[1]
  \STATE Let $t$ be the time
  \WHILE {(true)}
     \STATE \textbf{Execute in $t\%4=1$ steps} /*$QS_1$*/
     \STATE $\tilde{\pi} \leftarrow $ quick\_sort($U$) in decreasing order
     \STATE $C\leftarrow \tilde{\pi}^{-1}(\{1,2,...,k+c'\alpha\ln n\}) $ /*The constant $c'$ is determined in the proof*/
     \STATE \textbf{Execute in $t\%4=2$ steps} /*$QS_2$*/
     \STATE $\tilde{\pi}_C \leftarrow $  quick\_sort($C$) in decreasing order
     \STATE $P\leftarrow \tilde{\pi}_C^{-1}(\{1,2,...,k\})$
     \STATE \textbf{Execute in $t\%4=3$ steps} /*$QS_3$*/
     \STATE $\tilde{\pi}_P\leftarrow $ quick\_sort($P$) in decreasing order
     \STATE \textbf{Execute in $t\%4=0$ steps} /*Local-sort*/
     \STATE $\tilde{\pi}_k \leftarrow $  \emph{Local-sort}$(P,\tilde{\pi}_P, 4c+1)$/*The constant $c$ is determined in the proof*/
  \ENDWHILE
  \caption {\emph{Top-$k$-selection}}\label{topkorder}
  \end{algorithmic}
\end{algorithm}
\vspace{-0.5cm}

\begin{algorithm}[!h]
  \textbf{Input:} A set $P$; an order $\pi$ over $P$; an integer $c$  \\
  \textbf{Output:} $\tilde{\pi}$
  \begin{algorithmic}[1]
   \STATE  $m\leftarrow|P|$
    \STATE $B_1\leftarrow\pi^{-1}(\{1,2,...,c\})$ /* Define the first block */
    \STATE $\tilde{\pi}^{-1}(1)\leftarrow$ \emph{Maximum-Find}$(B_1)$
    \STATE $j=2$
    \WHILE {$(c+j-1 \leq m)$}
      \STATE $B_j \leftarrow (B_{j-1} \backslash \tilde{\pi}^{-1}(j-1)) \bigcup \pi^{-1}(c+j-1)$ /* Define the $j$-th block */
      \STATE $\tilde{\pi}^{-1}(j) \leftarrow$ \emph{Maximum-Find}$(B_j)$
      \STATE $j++$
    \ENDWHILE
    \STATE $B_e \leftarrow B_{j-1}$ /*Deal with the final block*/
    \WHILE { $|B_e| \geq 1$ }
      \STATE $\tilde{\pi}^{-1}(j) \leftarrow$ \emph{Maximum-Find}$(B_e)$
      \STATE $B_e \leftarrow B_e \backslash \tilde{\pi}^{-1}(j)$
      \STATE $j++$
    \ENDWHILE
  \caption{\emph{Local-sort}}\label{localsort}
 \end{algorithmic}
\end{algorithm}
\vspace{-0.3cm}
The main idea of Algorithm \ref{localsort} (\emph{Local-sort}) is to correct the order over $P$ block by block. Since block-by-block processing takes linear time, the errors can be corrected in time and little new errors will emerge during one processing period. Considering that the elements may move across blocks, it is necessary to make the blocks overlap. Actually, for each $j$, the element of the lowest rank in the $j$-th block is found, regarded as the $j$-th element of the final top-$k$ order, and removed from the block. The rest elements of the $j$-th block, together with the lowest-ranking element in $P$ (according to the latest order produced by $QS_3$) that has not yet been processed, forms the $(j+1)$-th block. The element of the lowest rank in each block is found by calling Algorithm \ref{largefind}, which runs one pass of sequential comparison. Both Algorithm \ref{localsort} and Algorithm \ref{largefind} are self-explained, so detailed explanation is omitted here.

\begin{algorithm}
  \textbf{Input:} $B$  \\
  \textbf{Output:} $u_{max}$
  \begin{algorithmic}[1]
    \STATE $u_{max} \leftarrow B(1)$
    \STATE $j=2$
    \WHILE {$( j \leq |B|)$}
      \IF { $u_{max} < B(j)$ }
        \STATE $u_{max} \leftarrow B(j)$
      \ENDIF
    \STATE $j++$
    \ENDWHILE
  \caption{\emph{Maximum-Find}}\label{largefind}
 \end{algorithmic}
\end{algorithm}

The following lemma will be used in analyzing performance of Algorithm \ref{topkorder}.
\def\pivotstatic{
In the \emph{consecutive-swapping model} with $\alpha=o(n)$, consider a run of the standard randomized quick-sort. If no swaps involve any element that is playing the role of a pivot, for any $u_i\in U$, the number of incorrectly ordered pairs $(u_i,u_j)$ is bounded by the number of swaps that involve $u_i$.
}
\begin{Lemma}\label{pivotstatic}
\pivotstatic
\end{Lemma}
\begin{proof}
This is in fact a by-product of the proof of Lemma 2.

Assume that no swaps involve any element that is playing the role of a pivot. In the proof of Lemma 2, we must have $S_2=\emptyset$. The lemma immediately follows.
\end{proof}

Now we analyze the performance of Algorithm \ref{topkorder}.

\def\corollaryeleven{
 Assume $\alpha=o(\frac{\sqrt{n}}{\ln{n}})$ and $k=O((\frac{n}{\alpha\ln{n}})^{1-\epsilon})$, where $\epsilon>0$. Let $\tilde{\pi}_k^t$ be the output of Algorithm \ref{topkorder} and $\pi_k^t$ be the true order over the largest $k$ elements at time $t$. For sufficiently large $t$, we have that:
 \begin{enumerate}
   \item If $k^2\alpha=o(n)$, $\Pr(\mathrm{KT}(\tilde{\pi}_k^t, \pi_k^t)=0)=1-o(1)$,
   \item If $k^2\alpha=\Theta(n)$, $\Pr(\mathrm{KT}(\tilde{\pi}_k^t, \pi_k^t)=0)=\Omega(1)$, and
   \item If $k^2\alpha=\omega(n)$, $\mathrm{KT}(\tilde{\pi}_k^t, \pi_k^t)=O(\frac{k^2\alpha}{n})$ with probability $1-o(1)$.
  \end{enumerate}
 }

 \begin{Theorem} \label{corollaryeeleven}
 \corollaryeleven
 \end{Theorem}
We first sketch the main idea of the proof. The proof consists of five steps. First, with high probability, the set $C$ produced by $QS_1$ includes all the largest $k$ elements during the next run of $QS_1$. Second, with high probability, the set $P$ produced by $QS_2$ exactly consists of the largest $k$ elements during the next run of $QS_2$. Third, with high probability, the true rank of any element remains close to that estimated by $QS_3$, during the next run of $QS_3$. Fourth, with high probability, the error of the order computed by \emph{Local-sort} is upper-bounded by the swaps during the run of \emph{Local-sort}. And fifth, proper upper bound of the swaps during a run of \emph{Local-sort} is presented. These steps immediately lead to the theorem.

\begin{proof}
We will show that the theorem holds at $t_0$, where $t_0$ is an arbitrary time step after $QS_1$ in Algorithm \ref{topkorder} has run twice.

Consider the last completed \emph{Local-sort} before $t_0$, which starts at $t_2$ and terminates at $t_1$. It is easy to see that $t_0-t_1=O(k)$ and $t_1-t_2=O(k)$. Also note that the input of \emph{Local-sort} at $t_2$ comes from the latest completed $QS_3$ before $t_2$, which starts at $t_4$ and finishes at $t_3$. By Lemma 1, with probability $1-o(1)$, $t_0-t_3=O(k\ln k)$ and $t_3-t_4=O(k\ln k)$. Likewise, the input of $QS_3$ at $t_4$ comes from the latest completed $QS_2$ before $t_4$, which starts at $t_6$ and finishes at $t_5$, and with probability $1-o(1)$, $t_0-t_5=O((k+\alpha\ln n)\ln(k+\alpha\ln n))$ and $t_5-t_6=O((k+\alpha\ln n)\ln(k+\alpha\ln n))$. Also, the input of $QS_2$ at $t_6$ comes from the latest completed $QS_1$, which starts at $t_8$ and terminates at $t_7$, and with probability $1-o(1)$, $t_0-t_7=O(n\ln n)$ and $t_7-t_8=O(n\ln n)$. The relationship between $t_0$ through to $t_8$ is illustrated in Fig. \ref{fig:times-general}.
\begin{figure}
\centering
\includegraphics[scale=0.4]{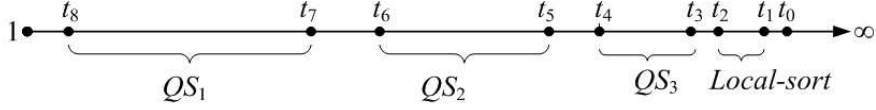}
\caption{Relationship between $t_0$ through to $t_8$}\label{fig:times-general}
\end{figure}

\textbf{Step 1}. We show that with probability $1-o(1)$, for all $u_i\in U$ and all $t\in [t_7, t_0]$, $|\pi^t(u_i)-\tilde{\pi}^{t_7}(u_i)|\leq c'\alpha\ln n$, where $c'$ is a constant. Due to the union bound, it is sufficient to show that for each $u_i\in U$, with probability $1-o(\frac{1}{n})$, $|\pi^t(u_i)-\tilde{\pi}^{t_7}(u_i)|\leq c'\alpha\ln n$ for all $t\in[t_7,t_0]$. Note that $|\pi^t(u_i)-\tilde{\pi}^{t_7}(u_i)|\leq |\pi^t(u_i)-\pi^{t_7}(u_i)|+|\pi^{t_7}(u_i)-\tilde{\pi}^{t_7}(u_i)|$, and $\Pr(|\pi^{t_7}(u_i)-\tilde{\pi}^{t_7}(u_i)|\leq c_1\ln n)=1-O(\frac{1}{n^2})$ for some constant $c_1$ by Lemma 2. Hence, we just need to show that $\Pr(\forall t\in [t_7, t_0], |\pi^t(u_i)-\pi^{t_7}(u_i)|\leq c_2\ln n)=1-O(\frac{1}{n^2})$ where $c_2$ is a constant, and let $c'=c_1+c_2$. Actually this follows from Chernoff bound and the fact that during $[t_7,t_0]$, the rank of $u_i$ changes $O(\ln n)$ in expectation.

As a result, event $\mathcal{E}_1$ happens with probability $1-o(1)$, where $\mathcal{E}_1$ means that the set $C$ produced at $t_7$ by $QS_1$, denoted as $C^{t_7}$, contains all the largest $k$ elements during $[t_7, t_0]$.

\textbf{Step 2}. We show that event $\mathcal{E}_2$ happens with probability $1-o(1)$, where $\mathcal{E}_2$ means that the set $P$ produced by $QS_2$ at $t_5$, denoted as $P^{t_5}$, exactly consists of the largest $k$ elements of $U$ during $[t_5, t_0]$. Considering the semantics of $QS_2$, it follows from two facts. On the one hand, $\mathcal{E}_1$ happens with probability $1-o(1)$. On the other hand, event $\mathcal{E}_3$ happens with probability $1-O(\frac{(k+\alpha\ln n)\alpha\ln(k+\alpha\ln n)}{n})=1-o(1)$, where $\mathcal{E}_3$ means that the $k$-th element at time $t_6$ doesn't swap during $[t_6, t_0]$.

Hence, hereunder we'll assume that $\mathcal{E}_1$, $\mathcal{E}_2$, and $\mathcal{E}_3$ all happen.

\textbf{Step 3}. Let the set $P$ at $t_3$ be $P^{t_3}$, where $P^{t_3}=P^{t_5}$, and $\tilde{\pi}_P^{t_3}$ be the order of $P^{t_3}$. We now show that for a constant $c$, with probability $1-o(1)$, for all $u_j\in P^{t_3}$ and all $t\in [t_3, t_0]$, $|\pi^t(u_j)-\tilde{\pi}_P^{t_3}(u_j)|\leq c$. To see why, first consider $Y_i$, the number of times that the element $(\pi^{t_4})^{-1}(i)$ swaps during $[t_4, t_0]$. Since $t_0-t_4=O(k\ln k)$, $E[Y_i]=O(\frac{\alpha k\ln k}{n})$. By Chernoff bound, for any constant $c_3\geq e$, $\Pr(Y_i\geq c_3)\leq (E[Y_i])^{c_3}$. Because $k=O((\frac{n}{\alpha\ln{n}})^{1-\epsilon})$, there is a constant $0<\beta<1$ such that $\frac{\alpha k\ln k}{n}=O(k^{\beta-1})$. Arbitrarily choose a constant $c_4>\max\{(1-\beta)^{-1}, e\}$, and we have $\Pr(\exists i\in [k], Y_i\geq c_4)\leq k(O(\frac{k\alpha\ln k}{n}))^{c_4}=O(k^{1+(\beta-1)c_4})=o(1)$. Consequently, $\Pr(\forall i\in [k], Y_i<c_4)=1-o(1)$. On the other hand, since $t_3-t_4=O(k\ln k)$, with probability $1-O(\frac{\alpha k\ln k}{n})=1-o(1)$, no element swaps when it is a pivot in $QS_3$ during $[t_4, t_3]$. By Lemma \ref{pivotstatic}, with probability $1-o(1)$, for all $u_j\in P^{t_3}$, $|\pi^{t_3}(u_j)-\tilde{\pi}_P^{t_3}(u_j)|\leq c_4$. As a result, given $c=2c_4$, with probability $1-o(1)$, for all $u_j\in P^{t_3}$ and all $t\in [t_3,t_0]$, $|\pi^t(u_j)-\tilde{\pi}_P^{t_3}(u_j)|\leq |\pi^t(u_j)-\pi^{t_3}(u_j)|+|\pi^{t_3}(u_j)-\tilde{\pi}_P^{t_3}(u_j)|\leq c$.

\textbf{Step 4}. Focus on the behavior of \emph{Local-sort} during $[t_2,t_1]$. Assume that during $[t_2,t_1]$, no element swaps while it is $u_{max}$. We claim that for any $u_i, u_j\in P^{t_3}$, if $\pi^{t}(u_i)<\pi^{t}(u_j)$ for all $t\in[t_2,t_1]$, then $\tilde{\pi}_k^{t_1}(u_i)<\tilde{\pi}_k^{t_1}(u_j)$. The claim can be proved in three cases. Assume that $\pi^{t}(u_i)<\pi^{t}(u_j)$ for all $t\in[t_2,t_1]$. Let $r_i\triangleq\tilde{\pi}_P^{t_3}(u_i)$ and $r_j\triangleq\tilde{\pi}_P^{t_3}(u_j)$.
\begin{itemize}
  \item \textbf{Case 1}: $u_i, u_j\in B_l$ for some $l$. Let $l_0$ be the largest such $l$. This implies that $u_{max}$ of $B_{l_0}$ is either $u_i$ or $u_j$. Since no element swaps while it is $u_{max}$, the $u_{max}$ computed by Local-sort cannot be $u_j$. This means that $u_{max}$ of $B_{l_0}$ is $u_i$ and $\tilde{\pi}_k^{t_1}(u_i)=l_0<\tilde{\pi}_k^{t_1}(u_j)$.
  \item \textbf{Case 2}: $r_i<r_j$ and there is no $l$ such that $u_i, u_j\in B_l$. There must be some $r_i-4c \leq l<r_j-4c$ such that  $u_{max}$ of $B_l$ is $u_i$. Hence, $\tilde{\pi}_k^{t_1}(u_i)=l<r_j-4c\leq \tilde{\pi}_k^{t_1}(u_j)$.
  \item \textbf{Case 3}: $r_i>r_j$ and there is no $l$ such that $u_i, u_j\in B_l$. We have $r_i-r_j<r_i-\pi^{t}(u_i)+\pi^{t}(u_j)-r_j$ for any $t\in[t_2,t_1]$, so $r_i-r_j<2c$. For any integer $m$, let $V_m= (\tilde{\pi}_P^{t_3})^{-1}(\{1,2,...,m+4c\})$. For all $r_j-4c\leq l<r_i-4c$, since $B_l\subseteq V_l$ and $|V_l\setminus V_{r_j-6c-1}|\leq 4c$, one gets $|B_l\bigcap V_{r_j-6c-1}|\geq 1$. For any $u'\in B_l\bigcap V_{r_j-6c-1}$, for all $t\in[t_2,t_1]$, $\pi^{t}(u')\leq \tilde{\pi}_P^{t_3}(u')+c\leq r_j-c-1 <\pi^{t}(u_j)$. By the proof of case 1 (with $u_i$ replaced by $u'$), $u_{max}$ of $B_l$ is not $u_j$ for any $r_j-4c\leq l<r_i-4c$. This means that $u_j,u_i\in B_{r_i-4c}$, which is a contradiction.
\end{itemize}

Since the probability that $u_{max}$ doesn't swap during $[t_2,t_1]$ is $1-O(\frac{k\alpha}{n})$, the above claim means that $\Pr(\mathrm{KT}(\tilde{\pi}_k^{t_1},\pi_k^{t_1})\leq Y)=1-o(1)$, where $Y$ is the number of swaps occurred in $P$ during $[t_2,t_1]$.

\textbf{Step 5}. Since $\Pr(Y=0)=(1-\frac{k}{n})^{O(k)\alpha}$, $\Pr(Y=0)=1-o(1)$ if $k^2\alpha=o(n)$, and $\Pr(Y=0)=\Omega(1)$ if $k^2\alpha=\Theta(n)$. When $k^2\alpha=\omega(n)$, note that $E[Y]=O(\frac{k^2\alpha}{n})$, so $Pr(Y=O(\frac{k^2\alpha}{n}))=1-o(1)$ by Chernoff bound. Actually, we still have these results if $[t_2, t_0]$ is considered.

Altogether, the theorem holds.
\end{proof}
We also analyze the lower bounds of the performance of any top-$k$-selection algorithms.
\def\corollaryone{
If $\alpha=o(n)$. And if $k=\Omega(\sqrt{\frac{n}{\alpha}})$, for arbitrary $t>k$, no algorithms can estimate the order of the largest $k$ elements such that $\Pr(\tilde{\pi}_k^t\equiv \pi_k^t)=1-o(1)$.
}

\begin{Theorem} \label{corollaryoneone}
\corollaryone
\end{Theorem}
\begin{proof}
Let us focus on the special case where $k=c_1\sqrt{\frac{n}{\alpha}}$ for some constant $c_1$, since it trivially implies the general case $k=\Omega(\sqrt{\frac{n}{\alpha}})$.

Consider the time interval $I=[t-c_2\sqrt{\frac{n}{\alpha}}, t]$, where $c_2 < c_1$. Let $p_1$ be probability that only one pair swaps among the largest $k$ elements during the interval $I$. Then $p_1=\frac{(k-1)(c_2\sqrt{n\alpha})}{n-1}(1-\frac{k}{n-1})^{c_2\sqrt{n\alpha}-1}$. When $n$ approaches infinity, $p_1=c_1c_2 e^{-c_1c_2}$. During the interval $I$, since at each time step only one pair is observed by the algorithm, therefore at most $2c_2\sqrt{\frac{n}{\alpha}}$ elements are checked. Let $p_2$ be the probability that only one pair swaps among the largest $k$ elements during the interval $I$ and the pair is not observed. Then $p_2 \geq \frac{p_1(c_1\sqrt{n/\alpha}-1-4c_2\sqrt{n/\alpha})}{c_1\sqrt{n/\alpha}-1}=\frac{p_1(c_1-4c_2)}{c_1}$. Let $c_1=5c_2$, then $p_2\geq c_2^2 e^{-5c_2^2}$. Thus the algorithm cannot tell the largest $k$ elements with correct order at time step $t_2$ with probability at least $p_2$, implying the result.
\end{proof}
\def\theoremsix{
Assume that $k=o(n)$ and $k=\omega(\sqrt{n})$, $\alpha=O(1)$, and $t>k/8$. For any algorithm $A$ solving top-$k$-selection problem, let $\tilde{\pi}_k^t$ be the output of $A$ and $\pi_k^t$ be the true order of $A$ at time $t$. If $(\tilde{\pi}_k^t)^{-1}(\{1,2,...k\})=(\pi_k^t)^{-1}(\{1,2,...k\})$, then $\mathrm{KT}(\tilde{\pi}_k^t,\pi_k^t)=\Omega(\frac{k^2}{n})$ in expectation and with probability $1-o(1)$.
}

\begin{Theorem} \label{theoremsix}
\theoremsix
\end{Theorem}
\begin{proof}

This proof is inspired by and is similar to that of \cite[Theorem 1]{Anagnostopoulos}. The proof is presented here not only to make this paper self-contained, but also due to subtle differences.

For convenience we assume that $\alpha=1$. And the proof can be modified slightly to prove the case where $\alpha >1$.

Consider the time interval $I=[t-\frac{k}{100},t]$. Let $X$ be the number of times that swaps occur among the largest $k$ elements. Then $E[X]=\Theta(\frac{k^2}{n})$. By Chernoff bound, we have that $X=\Theta(\frac{k^2}{n})$ with probability $1-o(1)$.

We use the idea of deferred decisions in\cite{Anagnostopoulos}. In the \emph{consecutive-swapping model}, one random pair of consecutive elements is chosen to swap in every step. The process is almost equivalent to the following process: two disjoint pairs of consecutive elements are picked uniformly and randomly and then one of these two pairs is selected to swap at random. This process is called nature's decision. The idea of deferred decisions is that we fix all the decisions that includes at least one of the elements observed by the algorithm and defer the rest. Since each swap will influence what pairs are for future swaps, we also need to fix both swaps that overlap. Therefore, the deferred decisions are a random set of disjoint pairs that is not involved with any elements observed by the algorithm or nature's fixed decisions so far. An element is said to be \emph{touched} if it is observed by the algorithm or is involved in any fixed decisions.

Initially, both the set of touched elements and the set of deferred decisions are empty. In every time step, a pair of elements is selected by the algorithm to compare. We mark each of these elements touched if it is previously untouched. For each deferred decision, we flip a coin with the suitable probability to determine if the decision contains that element newly marked as touched. If so, the decision is fixed and one pair in the decision is picked to swap with the appropriate probability. The decision may involve other previously untouched elements and we mark all those elements touched. And again, we continue to determine if any of the deferred decisions involves any of those elements. The process continues until we check  all the newly marked touched elements and all the deferred decisions. Then the comparison of the pair picked by the algorithm is answered. Next, the nature makes a new decision. We flip a coin to determine whether the new decision involves any touched elements. If it does, we mark the decision as fixed,  update the set of touched elements, and iterate as before. Also, a coin is flipped to determine whether the new decision overlaps with any of the deferred decisions. If it does, we fix both of the decisions and choose one pair in each decision to swap with appropriate probability.

At the end of interval $I$, the set of all the touched elements is of size at most $\frac{6k}{100}$. Thus at least $(1-\frac{12}{100})k$ pairs don't include any touched elements. Therefore, the number of deferred decisions is at least $(1-\frac{12}{100})^2X > \frac{1}{2}X$. Thus $\mathrm{KT}(\tilde{\pi}_k^t,\pi_k^t)=\Omega(\frac{k^2}{n})$ in expectation and with probability $1-o(1)$.

\end{proof}
By Theorems \ref{corollaryoneone} and \ref{theoremsix}, we know that $\Theta(n)$ is the critical point of $k^2\alpha$, and that \textbf{generally it is impossible to improve Algorithm \ref{topkorder}} even if $k^2\alpha=\omega(n)$. The term \emph{critical point} means that top-$k$-selection problem can be solved error-free with probability $1-o(1)$ if and only if $k^2\alpha=o(n)$.

\section{Gaussian-swapping Model}\label{sec:Guassian}
\vspace{-0.2cm}

This section is devoted to extending the algorithms for the consecutive-swapping model to the Gaussian-swapping model. We focus on the special case where $\alpha$ is a constant, and still assume that at each time step only one pair of elements can be queried.

\def\ddefinition{In the \emph{Gaussian-swapping Model} $U$, at any time $t$, define \emph{$u_i$-neighbourhood }$\triangleq\{u_j \in U: |\pi^t(u_j)-\pi^t(u_i)|\leq 4\sqrt{\ln n}\}$.
}
\begin{Definition}
\ddefinition
\end{Definition}
\def\Gaussianquick{
In the \emph{Gaussian-swapping Model}, the standard randomized quick-sort can terminate in $O(n\ln n)$ time in expectation and with probability $1-o(1)$.
}
\begin{Lemma}  \label{gaussianquick}
\Gaussianquick
\end{Lemma}
\begin{proof}
This proof is similar to Lemma \ref{quicksort}. Lemma \ref{quicksort} relies on the key fact that with probability $1-o(1)$, a pivot that is good at the time it is chosen will divide an array of size $s$ into two sub-arrays each having size at least $\frac{\gamma s}{2}$. Though this fact is an easy observation in the consecutive-swapping model, it is harder in the Gaussian-swapping model. In the following, we will only prove this fact in detail, since the other part of the proof of Lemma \ref{quicksort} can be directly used.

In the Gaussian-swapping model, $f(d)=\beta e^{-d^2/2}$. We define $p_j$ to be the probability that one pair with distance $j$ swaps and $q_i$ be the probability that the element of rank $i$ is chosen to swap. We have
$$q_i=\sum_{j=1}^{i-1}p_j+
\sum_{j=1}^{n-i}p_j.$$
Since $p_j=\frac{\beta e^{-j^2/2}}{n-j}$, it's easy to see that $p_j$ decreases as $j$ increases. Hence, $q_1$ is the minimum and $q_{\lfloor{\frac{n+1}{2}}\rfloor}$ is the maximum and $2q_1 > q_{\lfloor{\frac{n+1}{2}}\rfloor}$. Thus, $q_i=\Theta(\frac{1}{n})$.

Since
\begin{align*}
\sum_{i\geq 4\sqrt{\ln n}} f(i) & < \frac{1}{n^8}\times n        \\
           &= \frac{1}{n^7},
\end{align*}
with probability $1-o(\frac{1}{n^5})$, pairs with distance greater than $4\sqrt{\ln n}$ do not swap in a run of quick-sort. Our following proof is based on this assumption.

Now let's consider one process that a pivot $u_i$ is chosen to partition the array of size $s$. Suppose that $u_i$ is good when it is chosen. This means that if no swaps happen while $u_i$ is a pivot, the array will be divided into two sub-arrays $A$ and $B$ each having at least $\gamma s$ elements. Let $A'$ and $B'$ be the two sub-arrays that are actually obtained. Let $C=B'\backslash B$. The elements of $C$ must come from the following two types of events:
\begin{enumerate}
\item The pivot $u_i$ swaps with another element. Each such event contributes at most $4\sqrt{\ln n}$ elements to $C$.

\item A pair within $u_i$-neighbourhood swaps. Each such event contributes at most $1$ element to $C$.
\end{enumerate}
Let $X$ be the number of elements of $C$ resulting from the first type of events. Then $E[X]=O(\frac{s\sqrt{\ln n}}{n})$. By Markov inequality, $\Pr(X\geq \frac{\gamma s}{4})=o(1)$.

Let $Y$ be the number of elements of $C$ resulting from the second type of events. Since elements with distance greater than $4\sqrt{\ln n}$ do not swap in a run with probability $1-o(1)$, $E[Y]=O(\frac{s \sqrt{\ln n}}{n})$. Therefore, by Markov inequality, $\Pr(Y\geq \frac{\gamma s}{4})=o(1)$.

Altogether, $\Pr(|C|\geq \frac{\gamma s}{2})=o(1)$, meaning that $\Pr(|A'|\geq \frac{\gamma s}{2}, |B'|\geq \frac{\gamma s}{2})=1-o(1)$.
\end{proof}

\def\rremark1{
 In the \emph{Gaussian-swapping model}, consider a run of the standard randomized quick-sort from time $t_0$ to $t_1$. For any $u_i\in U$, the number of incorrectly ordered pairs $(u_i,u_j)$ is $O(\ln^{1.5}n)$ with probability $1-O(\frac{1}{n^3})$. Specifically, with probability $1-O(\frac{1}{n^2})$, $|\pi^{t_1}(u_l)-\tilde{\pi}^{t_1}(u_l)|=O(\ln^{1.5}n)$ for any $l$.
}
\begin{Lemma}  \label{rremark1}
\rremark1
\end{Lemma}

\begin{proof}

By Lemma \ref{gaussianquick}, the running time is $O(n\ln n)$ in expectation and with probability $1-o(1)$. We will use two facts in the proof of Lemma \ref{gaussianquick}. First, in each time step the probability that any element is chosen to swap is $\Theta(\frac{1}{n})$. Second, pairs with distance greater than $4\sqrt{\ln n}$ will not swap with probability $1-o(1)$ during the run.

Similar to the proof of \cite[Lemma 6]{Anagnostopoulos}, we partition the set of incorrectly ordered pairs into two sets, $A$ and $B$.

$A=\{u_j: u_i <_{\tilde{\pi}^{t_1}} u_j,u_i >_{{\pi}^{t_1}} u_j, \exists t\in [t_0, t_1):u_i <_{{\pi}^{t}} u_j  \} \bigcup \{u_j: u_i >_{\tilde{\pi}^{t_1}} u_j,u_i <_{{\pi}^{t_1}} u_j, \exists t\in [t_0, t_1):u_i >_{{\pi}^{t}} u_j  \}$

$B=\{u_j:u_i <_{\tilde{\pi}^{t_1}} u_j, \forall t\in [t_0, t_1]: u_i >_{{\pi}^{t}} u_j \} \bigcup \{u_j:u_i >_{\tilde{\pi}^{t_1}} u_j, \forall t \in [t_0, t_1]: u_i <_{{\pi}^{t}} u_j \}   $

$A$ can be covered by two cases:

\textbf{Case 1}: swaps involving with $u_i$,

\textbf{Case 2}: swaps among $u_i$-neighbourhood not involving with $u_i$.

 In each time step, a swap of case 1 can contribute at most $4\sqrt{\ln n}$ elements to set $A$. Recall that in each time step, the probability that $u_i$ is chosen to swap is $\Theta(\frac{1}{n})$. By the Chernoff bound, the number of such elements is bounded by $c_1\ln^{1.5}n$ in expectation and with probability $1-O(\frac{1}{n^3})$ for some constant $c_1$.

On the other hand, in each time step, a swap of case 2 results in at most 2 elements to set $A$. Recall that in every time step, the probability that a swap occurs among $u_i$-neighborhood is $O(\frac{\sqrt{\ln n}}{n})$. Using the Chernoff bound, the number of such elements is bounded by $c_2\ln^{1.5}n$ in expectation and with probability $1-O(\frac{1}{n^3})$ for some constant $c_2$.

For set $B$, consider the case $u_i < u_j$ according to the true order during $[t_0, t_1]$( the case $u_i>u_j$ is similar). In order for element $u_i$ and $u_j$ to be incorrectly ordered, there must be some pivot $u_l$ such that $u_j<u_l$ and $u_l<u_i$ when the comparisons occur, respectively. So, the relative rank of $u_l$ and $u_i$, $u_l$ and $u_j$ has to change at least once. Only the following two kinds of swaps can contribute to set $B$:

\textbf{Case 3}: swaps involving $u_l$,

\textbf{Case 4}: swaps among $u_l$-neighbourhood not involving $u_l$.

Since other swaps beyond case 3 and case 4 cannot make any element $u_j$ such that the relative rank of $u_j$ and $u_l$ changes, thus only case 3 and case 4 contribute to set $B$.

For case 3, each time step results in at most $4\sqrt{\ln n}$ elements to set $B$. Let $X_l$ be the number of steps that element $u_l$ is a pivot and $Y_l$ be the number of steps that $u_l$ swaps when $u_l$ is a pivot. Let $P_i \subseteq U$ be the set of elements that acted as pivots along the path in the quick-sort tree of element $u_i$. Then by properties of the quick-sort algorithm, we have
$$E[\sum_{l:u_l\in P_i}X_l]=O(n),$$ and
$$ \sum_{l:u_l\in P_i}X_l \leq c_5n\ln n, $$
with probability $1-o(1)$, for a constant $c_5$.

We define $Z_l=\sum_{l:u_l\in P_i}Y_l| X_l$. Since $Y_l \sim Binomial(X_l, \Theta(\frac{1}{n}))$, $E[E[Z_l]]\leq E[\sum_{l:u_l\in P_i}X_l\Theta(\frac{1}{n})]=O(1)$. By Chernoff bound, $Z_l=O(\ln n)$ with probability $1-o(1)$. Therefore, case 3 produces at most $c_3\ln^{1.5} n$ such elements with probability $1-O(\frac{1}{n^3})$, for some constant $c_3$.

For case 4, each time step produces at most 2 elements whose relative rank with $u_l$ changes with probability $\Theta(\frac{\sqrt{\ln n}}{n})$. Let $W_l$ be the number of such swaps when $u_l$ is a pivot. We have that $E[\sum_{l:u_l\in P_i}W_l| X_l]=O(\frac{\sqrt{\ln n}}{n}X_l)=O(\sqrt{\ln n})$. By Chernoff bound, case 4 leads to at most $c_4\ln^{1.5}n$ such elements with probability $1-O(\frac{1}{n^3})$.

Let $c=c_1+c_2+c_3+4_4$. Then we get the first part of the lemma.

By union bound, for every $l$, $|\pi^{t_1}(u_l)-\tilde{\pi}^{t_1}(u_l)|\leq c\ln^{1.5}n$ holds with probability at least $1-O(\frac{1}{n^2})$. The second part of the lemma is proven.

\end{proof}

\begin{algorithm}
 \textbf{Input:} A set of elements $U$\\
 \textbf{Output:} $\widetilde{T}$
 \begin{algorithmic}[1]
 \WHILE{(true)}
    \STATE \textbf{Execute in odd steps:} /*$QS_1$*/
    \STATE $\tilde{\pi} \leftarrow $ quick\_sort($U$) in decreasing order
    \STATE $L\leftarrow \tilde{\pi}^{-1}(\{1,2,...,k-c\ln^{1.5}n\})$ and $C\leftarrow \tilde{\pi}^{-1}(\{k-c\ln^{1.5} n+1,...,k+c\ln^{1.5}n\})$ /*The constant $c$ is determined in the proof*/
    \STATE \textbf{Execute in even steps:} /*$QS_2$ */
    \STATE $\tilde{\pi}_C\leftarrow $ quick\_sort($C$) in decreasing order
    \STATE $\widetilde{T}\leftarrow L \bigcup \tilde{\pi}_C^{-1}(\{1,2,...,c\ln^{1.5}n \}) $
 \ENDWHILE
 \caption{\emph{Gaussian-Top-$k$-set}} \label{fig:Gausstopk}
 \end{algorithmic}
\end{algorithm}

Algorithm \ref{GaussianTone} solves the top-$k$-set problem in the \emph{Gaussian-swapping model}. It is similar to Algorithm \ref{fig:topk}, only different in the sizes of the sets $L$ and $C$. The difference is due to the different variations of an element's rank during a period of $O(n\ln n)$, as shown in Lemma \ref{range} and Lemma \ref{rremark1}.

The following theorem is a counterpart of Theorem \ref{The:theoremone}. It shows that Algorithm \ref{fig:Gausstopk} is also error-free with high probability, in spite of the complicated data evolving model. But it is a little weaker than Theorem \ref{The:theoremone}, due to the restriction $\alpha=\Theta(1)$ which is far smaller than $o(\frac{\sqrt{n}}{\ln{n}})$. The proof is similar to that of Theorem \ref{The:theoremone}.
\def\rr{
For any $1\leq k\leq n$, we have $\Pr(\widetilde{T}_t = (\pi^t)^{-1}(\{1,2,...,k\}))=1-o(1)$, where $\widetilde{T}_t$ is the output of Algorithm \ref{fig:Gausstopk}, $\pi^t$ is the true order at time $t$, and $t$ is such that the algorithm has run at least once.
}
\begin{Theorem} \label{GaussianTone}
\rr
\end{Theorem}

\begin{proof}

Consider a run of $QS_1$, which starts at $t_0$ and ends at $t_1$. By Lemma \ref{gaussianquick}, we have $t_1-t_0=O(n\ln{n})$ in expectation and with probability $1-o(1)$. By Lemma~\ref{rremark1}, we have that $|\pi^{t_0}(u_l)-\tilde{\pi}^{t_0}(u_l)|\leq c_1\ln^{1.5}n$ with probability $1-O(\frac{1}{n^2})$ for every $u_l\in U$ and some constant $c_1$. During $[t_0,t_1]$, the rank of an element $u_l$ changes less than $c_2\ln^{1.5}n$ with probability $1-o(\frac{1}{n^2})$ for some constant $c_2$. Hence we have $|\pi^t{(u_l)}-\tilde{\pi}^{t_0}(u_l)|\leq (c_1+c_2)\ln^{1.5}n$ with probability $1-o(1)$ at any time $t\in [t_0,t_1]$ and any $u_l\in U$. Letting $c=c_1+c_2$, we have $|\pi^t{(u_l)}-\tilde{\pi}^{t_0}(u_l)|\leq c\ln^{1.5}n$ with probability $1-o(1)$.

Note that $L$ contains all the elements $u_i$ such that $\tilde{\pi}^{t_0}(u_i)\leq k-c\ln^{1.5}n$. Then for any $t\in[t_0,t_1]$ and any $u_i\in L$, we have that $\pi^t(u_i)\leq k$ with probability $1-o(1)$. Consider the set $R=U\backslash( L\bigcup C)=\{u_l:\tilde{\pi}^{t_0}(u_l)\geq k+c\alpha\ln n+1\}$. Then for any $t\in[t_0,t_1]$ and any $u_i\in R$, we have that $\tilde{\pi}^t(u_i)>k$ with probability $1-o(1)$. Therefore, $L\bigcup C$ contains all the elements whose true rank is no more than $k$ during $[t_0,t_1]$, with probability $1-o(1)$.

In line 6, quick\_sort$(C)$ requires time $O(\ln^{1.5}n\ln\ln n)$ with probability $1-o(1)$. Let the element of rank $k$ at the beginning of line 6 be $\tilde{u}$. If those elements in $\tilde{u}$-neighbourhood do not swap during the execution of line 6, the algorithm can always return the correct set of the largest $k$ elements at the end of quick\_sort$(C)$. This reason lies in two aspect. On the one hand, any pair whose distance is greater than $4\sqrt{\ln n}$ will not swap during a full execution of line 6 with probability $1-o(1)$. On the other hand, the rank of any element $u_i\in C$ remain smaller than $k-4\sqrt{\ln n}$ if it is smaller than $k-4\sqrt{\ln n}$ at the beginning and $\tilde{u}$-neighbourhood do not swap during the execution of line 6. Note that the probability that $\tilde{u}$-neighbourhood do not change places is $1-O(\frac{\ln^3 n\ln\ln n}{n})=1-o(1)$. During the next round of sorting $C$(before a new set of the largest $k$ elements is computed), the set of the largest $k$ elements can remain the same with probability $1-O(\frac{\ln^3 n\ln\ln n}{n})=1-o(1)$, implying the result.
\end{proof}

\begin{algorithm}
 \textbf{Input:} A set of elements $U$\\
 \textbf{Output:} $\tilde{\pi}_k$
 \begin{algorithmic}[1]
 \STATE Let $t$ be the time
 \WHILE{(true)}
    \STATE \textbf{Execute in $t\%3=1$ steps} /*$QS_1$*/
    \STATE $\tilde{\pi}\leftarrow$ quick\_sort$(U)$ in decreasing order
    \STATE $C\leftarrow \tilde{\pi}^{-1}(\{1,2,...,k+c'\ln^{1.5}n\})$ /*The constant $c'$ is determined in the proof*/
    \STATE \textbf{Execute in $t\%3=2$ steps} /*$QS_2$*/
    \STATE $\tilde{\pi}_C\leftarrow $ quick\_sort$(C)$ in decreasing order
    \STATE $P \leftarrow  \tilde{\pi}_C^{-1}(\{1,2,...,k\})$ and $\tilde{\pi}_P(u_i)\leftarrow \tilde{\pi}_C(u_i)$ for all $u_i\in P$
    \STATE \textbf{Execute in $t\%3=0$ steps} /*Local-sort*/
    \STATE $\tilde{\pi}_k \leftarrow$ \emph{Local-sort}$(P, \tilde{\pi}_P, 4c\sqrt{\ln n}+1)$ /*The constant $c$ is determined in the proof*/
 \ENDWHILE
 \caption{\emph{Gaussian-Top-$k$-selection}} \label{algorithm6}
 \end{algorithmic}
\end{algorithm}

Now we present Algorithm \ref{algorithm6}. The algorithm is designed to solve the top-$k$-selection problem. Based on the same basic idea, it is similar to Algorithm \ref{topkorder}, with difference only in three aspects:
\begin{enumerate}
  \item The size of the set $C$ is $k+O(\ln^{3/2}n)$ rather than $k+O(\alpha\ln n)$ in Algorithm \ref{topkorder}. This is due to the $O(\ln^{3/2}n)$ ranking error of randomized quick-sort in the Gaussian-swapping model, by Lemma \ref{rremark1}.
  \item It does not need $QS_3$ of Algorithm \ref{topkorder}. The reason lies in the restriction that $\alpha=\Theta(1)$ rather than $o(\frac{\sqrt{n}}{\ln{n}})$. In fact, if focusing on the case $\alpha=\Theta(1)$, $QS_3$ can also be removed from Algorithm \ref{topkorder} without compromising the performance of Algorithm \ref{topkorder}.
  \item The third argument of \emph{Local-sort} is $4c\sqrt{\ln n}+1$ rather than $4c+1$, meaning that longer blocks are used in locally correcting the sorting errors. This is because that pairs with distance $\Theta(\sqrt{\ln n})$ can occur with high probability within time $O(n\ln n)$ in the Gaussian-swapping model.
\end{enumerate}

The next theorem shows the performance of Algorithm \ref{algorithm6}. The proof is similar to that of Theorem \ref{corollaryeeleven}.

\def\rrrrr{Assume $k=O((\frac{n}{\ln{n}})^{1-\epsilon})(\epsilon>0)$. Let $\tilde{\pi}_k^t$ be the output of \emph{Gaussian-Top-$k$-selection} and $\pi_k^t$ be the true order over the largest $k$ elements at time $t$. For sufficiently large $t$, we have that:
\begin{enumerate}
\item If $k=o(\frac{\sqrt{n}}{\ln^{0.25}n})$, $\Pr(\tilde{\pi}_k^t\equiv \pi_k^t)=1-o(1)$,
\item If $k=\Theta(\frac{\sqrt{n}}{\ln^{0.25}n})$, $\Pr(\tilde{\pi}_k^t\equiv \pi_k^t)=\Omega(1)$, and
\item If $k=\omega(\frac{\sqrt{n}}{\ln^{0.25}n})$, $\Pr(\mathrm{KT}(\tilde{\pi}_k^t,\pi_k^t)=O(\frac{k^2\ln n}{n}))=1-o(1)$.
\end{enumerate}
}
\begin{Theorem} \label{rrmakr5}
\rrrrr
\end{Theorem}

We first sketch the main idea of the proof. Similar to that of Theorem \ref{corollaryeeleven}, it consists of five steps. First, with high probability, the set $C$ produced by $QS_1$ includes all the largest $k$ elements during the next run of $QS_1$. Second, with high probability, the set $P$ produced by $QS_2$ exactly consists of the largest $k$ elements during the next run of $QS_2$. Third, with high probability, the true rank of any element remains close to that estimated by $QS_2$, during the next run of $QS_2$. Fourth, with high probability, the error of the order computed by \emph{Local-sort} is upper-bounded by the swaps multiplied by $O(\sqrt{\ln n})$ during the run of \emph{Local-sort}. And fifth, proper upper bound of the swaps during a run of \emph{Local-sort} is presented. These steps immediately lead to the theorem.
\begin{proof}

We'll also use two results in the proof of Lemma \ref{gaussianquick}. First, each time the probability that any element is chosen to swap is $\Theta(\frac{1}{n})$. Second, pairs with distance greater than $4\sqrt{\ln n}$ won't swap with probability $1-o(1)$ during $[t_6, t_0]$.

We will show that the theorem holds at $t_0$, where $t_0$ is an arbitrary time step after $QS_1$ in Algorithm \ref{algorithm6} has run twice.

Consider the last completed \emph{Local-sort} before $t_0$, which starts at $t_2$ and terminates at $t_1$. It is easy to see that $t_0-t_1=O(k\sqrt{\ln n})$ and $t_1-t_2=O(k\sqrt{\ln n})$. Also note that the input of \emph{Local-sort} at $t_2$ comes from the latest completed $QS_2$ before $t_2$, which starts at $t_4$ and terminates at $t_3$. By Lemma \ref{gaussianquick}, with probability $1-o(1)$, $t_3-t_4=O(T)$ and $t_0-t_3=O(T)$, where $T=O(\ln^{1.5} n\ln\ln n)$ if $k=o(\ln^{1.5} n)$ and $T=O(k\ln k)$ otherwise. Likewise, the input of $QS_2$ at $t_4$ comes from the latest completed $QS_1$ before $t_4$, which starts at $t_6$ and finishes at $t_5$, and with probability $1-o(1)$, $t_0-t_5=O(n\ln n)$ and $t_5-t_6=O(n\ln n)$. The relation from $t_6$ to $t_0$ is illustrated in Fig. \ref{fig:times-Gaussian}.

\textbf{Step 1}. We show that with probability $1-o(1)$, for all $u_i\in U$ and all $t\in [t_5,t_0]$, $|\pi^t(u_i)-\tilde{\pi}^{t_5}(u_i)|\leq c'\ln^{1.5}n$, where $c'$ is a constant.

Note that $|\pi^t(u_i)-\tilde{\pi}^{t_5}(u_i)|\leq |\pi^t(u_i)-\pi^{t_5}(u_i)|+|\pi^{t_5}(u_i)-\tilde{\pi}^{t_5}(u_i)|$, and $\Pr(|\pi^{t_5}(u_i)-\tilde{\pi}^{t_5}(u_i)|\leq c_1\ln^{1.5}n)=1-o(\frac{1}{n^2})$ for some constant $c_1$ by Lemma \ref{rremark1}. From Chernoff bound and the fact that during $[t_5, t_0]$, the rank of $u_i$ changes $O(\ln^{1.5} n)$ in expectation, $\Pr(|\pi^t(u_i)-\pi^{t_5}(u_i)|\leq c_2\ln^{1.5} n)=1-o(\frac{1}{n^2})$. Hence, due to Union bound, $\Pr(\forall t\in[t_5,t_0] \wedge \forall u_i\in U, |\pi^t(u_i)-\tilde{\pi}^{t_5}(u_i)|\leq c'\ln^{1.5} n)=1-o(\frac{1}{n})$, where $c'=c_1+c_2$.

As a result, event $\mathcal{E}_1$ happens with probability $1-o(1)$, where $\mathcal{E}_1$ means that the set $C$ produced at $t_5$ by $QS_1$, denoted as $C^{t_5}$, contains all the largest $k$ elements during $[t_5,t_0]$.

\textbf{Step 2}. We show that event $\mathcal{E}_2$ happens with probability $1-o(1)$, where $\mathcal{E}_2$ stands for the event that the set $P$ produced by $QS_2$ at $t_3$, denoted as $P^{t_3}$, exactly consists of the largest $k$ elements of $U$ during $[t_3,t_0]$.

Considering the semantics of $QS_2$, it follows from two facts. On the one hand, $\mathcal{E}_1$ happens with probability $1-o(1)$. One the other hand, event $\mathcal{E}_3$ happens with probability $1-O(\frac{T\sqrt{\ln n}}{n})=1-o(1)$, where $\mathcal{E}_3$ means that the neighbourhood of $k$-th element at time $t_4$ don't swap during $[t_4,t_0]$.

Therefore, hereunder we'll assume that $\mathcal{E}_1,\mathcal{E}_2$ and $\mathcal{E}_3$ all happen.

Now we prove that the theorem holds if $k=o(\ln^{1.5}n)$. When $k=o(\ln^{1.5} n)$, $|C|=O(\ln^{1.5}n)$. Therefore, with probability $1-o(1)$, $t_0-t_4=O(\ln^{1.5}n\ln\ln n)$. When $t_0-t_4=O(\ln^{1.5}n\ln\ln n)$, no element of $C$ changes places during $[t_4,t_0]$ with probability $1-O(\frac{\ln^3 n\ln\ln n}{n})=1-o(1)$. If the order of $C$ does not change throughout $[t_4,t_0]$, the order over $P$ estimated by $QS_2$ at $t_3$ is exactly $\pi^t|_P$ for all $t\in [t_3, t_0]$. Then, the order over $P$ estimated by \emph{Local-sort} at $t_1$ is equal to $\pi^t|_P$ for $t\in[t_1,t_0]$, immediately leading to the theorem.

\begin{figure}
\centering
\includegraphics[scale=0.4]{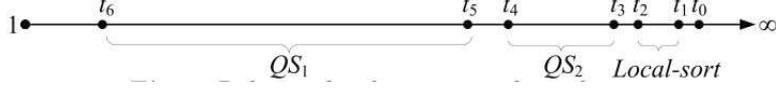}
\caption{Relationship from $t_6$ to $t_0$}\label{fig:times-Gaussian}
\end{figure}

As a result, in the rest of the proof, we assume that $k=\Omega(\ln^{1.5}n)$.

Let $\tilde{\pi}_C^{t_3}$ be the order of $C$ estimated by $QS_2$ at time $t_3$, and $\tilde{\pi}_P^{t_3}\triangleq \tilde{\pi}_C^{t_3}|_{P^{t_3}}$ be the induced order of $P^{t_3}$.

\textbf{Step 3}. We show that for a constant $c$, with probability $1-o(1)$, for all $u_j\in P^{t_3}$ and all $t\in [t_3,t_0]$, $|\pi^t(u_j)-\tilde{\pi}_P^{t_3}(u_j)|\leq c\sqrt{\ln n}$.

To see why, first consider $Y_i$, the number of times that the element $(\pi^{t_4})^{-1}(i)$ swaps during $[t_4,t_0]$. Since $t_0-t_4=O(k\ln k)$ in expectation and with probability $1-o(1)$, $E[Y_i]=O(\frac{k\ln k}{n})$. By Chernoff bound, for any constant $c_3\geq e$, $\Pr(Y_i\geq c_3)\leq (E[Y_i])^{c_3}$. Because $k=O((\frac{n}{\ln n})^{1-\epsilon})$, there is a constant $0<\beta<1$ such that $\frac{k\ln k}{n}=O(k^{\beta-1})$. Arbitrarily choose a constant $c_4>\max\{(1-\beta)^{-1}, e\}$ and we have $\Pr(\exists i\in [k], Y_i\geq c_4)\leq k(O(\frac{k\ln k}{n})^{c_4})=O(k^{1+(\beta-1)c_4})=o(1)$. Consequently, $\Pr(\forall i\in [k], Y_i<c_4)=1-o(1)$.

On the other hand, during $[t_4, t_3]$, for every pivot element $u_j$, if elements in $u_j$-neighbourhood do not change places when $u_j$ is a pivot, for any pair $(u_a,u_b)$ that $u_a < u_b$ always holds during $[t_4,t_3]$, $u_a<_{\tilde{\pi}_C^{t_3}} u_b$ also holds and such event happens with probability $1-O(\frac{k\ln k\sqrt{\ln n}}{n})=1-o(1)$.

For any element $u_i$ in the set $C$, we define $S_i$ as follows:

$S_i=\{u_j: u_i <_{\tilde{\pi}_C^{t_3}} u_j,u_i >_{{\pi}^{t_3}} u_j, \exists t\in [t_4, t_3):u_i <_{{\pi}^{t}} u_j  \} \bigcup \{u_j: u_i >_{\tilde{\pi}_C^{t_3}} u_j,u_i <_{{\pi}^{t_3}} u_j, \exists t\in [t_4, t_3):u_i >_{{\pi}^{t}} u_j  \}$

$|S_i|$ is bounded by the number of elements that once are in the $u_i$-neighbourhood at some time slot $t\in[t_4, t_0]$. Since during $[t_4, t_0]$, each element swaps at most $c_4$ times, there are at most $2c_4^2 \times 4\sqrt{\ln n}=8c_4^2\sqrt{\ln n}$ elements that are in the $u_i$-neighbourhood during $[t_4,t_0]$. Therefore, we have that $|\pi^t(u_j)-\tilde{\pi}_P^{t_3}(u_j)|\leq 8c_4^2\sqrt{\ln n}$. In the algorithm \ref{algorithm6}, let $c=8c_3^2$.

\textbf{Step 4.} We claim that with high probability, for any $u_i,u_j\in P^{t_3}$, $\tilde{\pi}_k^{t_1}(u_i)<\tilde{\pi}_k^{t_1}(u_j)$ if $\pi^t(u_i)<\pi^t(u_j)$ throughout $[t_2,t_1]$.

Since the probability that no swaps occur within $u_{max}$-neighbourhood during $[t_2,t_1]$ is $1-O(\frac{k\sqrt{\ln n}}{n})=1-o(1)$, we assume hereunder that no swaps occur within $u_{max}$-neighbourhood during $[t_2,t_1]$.

The claim can be proved in three cases. Assume that $\pi^t(u_i)<\pi^t(u_j)$ for all $t\in[t_2, t_1])$. Let $r_i\triangleq \tilde{\pi}_P^{t_3}(u_i)$ and $r_j\triangleq \tilde{\pi}_P^{t_3}(u_j)$.
\begin{itemize}
\item \textbf{Case 1}: $u_i,u_j\in B_l$. Let $l_0$ be the largest such $l$. This implies that either $u_{max}$ of $B_{l_0}$ is either $u_i$ or $u_j$. Since no element swaps while it is $u_{max}$, the $u_{max}$ computed by \emph{Local-sort} cannot be $u_j$. This means that $u_{max}$ of $B_{l_0}$ is $u_i$ and $\tilde{\pi}_k^{t_1}(u_i)=l_0<\tilde{\pi}_k^{t_1}(u_j)$.
 \item \textbf{Case 2}: $r_i<r_j$ and there is no $l$ such that $u_i,u_j\in B_l$. There must be some $r_i-4c\sqrt{\ln n}\leq l < r_j-4c\sqrt{\ln n}$ such that $u_{max}$ of $B_l$ is $u_i$. Hence, $\tilde{\pi}_k^{t_1}=l<r_j-4c\sqrt{\ln n}\leq \tilde{\pi}_k^{t_1}(u_j)$.
 \item \textbf{Case 3}: $r_i> r_j$ and there is no $l$ such that $u_i, u_j\in B_l$. We have $r_i-r_j<r_i-\tilde{\pi}^t(u_i)-r_j+\tilde{\pi}^t(u_j)$ for any $t\in [t_2,t_1]$, so $r_i-r_j<2c\sqrt{\ln n}$. For any integer $m$, let $V_m=(\tilde{\pi}_P^{t_3})^{-1}(\{1,2,...,m+4c\sqrt{\ln n}\})$. For all $r_j-4c\sqrt{\ln n}\leq l < r_i-4c\sqrt{\ln n}$, since $B_l\subseteq V_l$ and $|V_l\backslash V_{r_j-6c\sqrt{\ln n}-1}|\leq 4c\sqrt{\ln n}$, one gets $|B_l\bigcap V_{r_j-6c\sqrt{\ln n}-1}|\geq 1$. For any $u'\in B_l\bigcap V_{r_j-6c\sqrt{\ln n}-1}$ and for all $t\in [t_2, t_1]$, $\pi^t(u')\leq \tilde{\pi}_P^{t_3}(u')+c\sqrt{\ln n} \leq r_j-c\sqrt{\ln n}-1<\pi^t(u_j)$. By the proof of case 1 (with $u_i$ replaced by $u'$), $u_{max}$ of $B_l$ is not $u_j$ for any $r_j-4c\sqrt{\ln n}\leq l<r_i-4c\sqrt{\ln n}$. This means that there exists $l$ such that $u_i, u_j\in B_l$, which is a contradiction.
\end{itemize}

\textbf{Step 5}. The claim in Step 4 means that $\Pr(\mathrm{KT}(\tilde{\pi}_k^{t_1}, \pi_k^{t_1})\leq 4Y\sqrt{\ln n})=1-o(1)$, where $Y$ is the number of swaps occurring in $P$ during $[t_2, t_1]$. Since $\Pr(Y=0)=(1-\frac{k}{n})^{O(k\sqrt{\ln n})}$, $\Pr(Y=0)=1-o(1)$ if $k=o(\frac{\sqrt{n}}{\ln^{0.25} n})$, and $\Pr(Y=0)=\Omega(1)$ if $k=\Theta(\frac{\sqrt{n}}{\ln^{0.25} n})$. When $k=\omega(\frac{\sqrt{n}}{\ln^{0.25} n})$, note that $E[Y]=O(\frac{k^2\sqrt{\ln n}}{n})$, so $\Pr(\mathrm{KT}(\tilde{\pi}_k^{t_1},\pi_k^{t_1})=O(\frac{k^2\ln n}{n}))=1-o(1)$ by Chernoff bound. Actually, we still have these results if $[t_2,t_0]$ is considered.

Altogether, the theorem holds.
\end{proof}

As a counterpart of Theorem \ref{corollaryeeleven}, Theorem \ref{rrmakr5} differs in three aspects:
\begin{itemize}
  \item It only considers $\alpha=\Theta(1)$, rather than $\alpha=o(\frac{\sqrt{n}}{\ln{n}})$.
  \item The critical point of $k$ is $\Theta(\frac{\sqrt{n}}{\ln^{0.25}n})$, instead of $\Theta(\sqrt{\frac{n}{\alpha}})$.
  \item Beyond the critical point, the error bound is $O(\frac{k^2\ln n}{n})$, bigger than $O(\frac{k^2\alpha}{n})$.
\end{itemize}

Except for the Gaussian distribution, $d$ can also be determined by other discrete distributions, for example, $p(d)=\frac{\beta}{d^{\gamma}}$, where $\gamma$ is a constant and $\beta$ is a normalizing factor. When $\gamma$ is large enough (say, $\gamma>10$), the results similar to those in the Gaussian-swapping model can be obtained.

\vspace{-0.2cm}
\section{Conclusions}\label{sec:con}
\vspace{-0.2cm}

In this paper we completely solve the top-$k$-set problem for any $k\leq n$ by algorithmically identifying the exact largest $k$ objects with probability $1-o(1)$ at every time step in \emph{consecutive-swapping model} and \emph{Gaussian-swapping model}.

In \emph{consecutive-swapping model}, for the top-$k$-selection problem, it can also be exactly solved at each time step with probability $1-o(1)$ if $k=o(\sqrt{\frac{n}{\alpha}})$. And the upper bound is tight since when $k=\Omega(\sqrt{\frac{n}{\alpha}})$, no algorithms can estimate the correct order of the largest $k$ elements with probability $1-o(1)$ at any large enough time $t$. For the case where $k=\Theta(\sqrt{\frac{n}{\alpha}})$, Algorithm~\ref{topkorder} is designed which solves the top-$k$-selection problem error-free with a constant probability of success. When $\alpha$ is a constant and $k=\omega(\sqrt{n})$, we show that any algorithm solving this problem must have error lower bound $\Omega(\frac{k^2}{n})$ with probability $1-o(1)$. In fact, it is shown that the lower bound is tight if $k=O(({\frac{n}{\ln n}})^{1-\epsilon})$ for any $\epsilon>0$.

A number of problems remain open for the top-$k$-selection problem in the \emph{consecutive-swapping model}. For $\alpha=\omega(1)$, we don't show that if the error bound $O(\frac{k^2\alpha}{n})$ is tight when $k=\omega(\sqrt{\frac{n}{\alpha}})$. And for $\alpha=O(1)$, there exists a gap between $k=n$ and $k=O(({\frac{n}{\ln n}})^{1-\epsilon})$, where the error bound $\Omega(\frac{k^2}{n})$ has not yet shown to be tight. We conjecture that the bound keeps tight. Currently, for the special case $k=n$, the best known upper bound is $O(n\ln\ln n)$ \cite{Anagnostopoulos}. 
Another direction is to study the trade-off between the error and the accessibility to the data. For example, if at every time step, the algorithm is allowed to query the relative order of $\ln n$ objects, will the error vanish with high probability?

In the \emph{Gaussian-swapping model}, we obtained similar results for the top-$k$-selection problem. However, we have not yet obtained any non-trivial lower bound.

\end{document}